\documentclass[11pt,twoside]{article}

\usepackage{mathtools} 
\usepackage{mathrsfs} 
\usepackage[symbol]{footmisc}

\newcommand{\DD}{\mathcal{D}}

\newcommand{\FF}{\mathcal{F}}
\newcommand{\HH}{\mathscr H}

\newcommand{\hh}{\mathfrak{h}}

\usepackage[utf8]{inputenc} 
\usepackage{amsmath,amssymb,amsthm}
\usepackage{pstricks,pst-node,pst-coil,pst-plot,pstricks-add}
\usepackage{geometry,epsfig}
\usepackage{bbm}
\usepackage{cite}

\usepackage{cite}
\usepackage{paralist}
\usepackage{cleveref}
\usepackage{enumitem}
\usepackage{emptypage}

\DeclareGraphicsExtensions{.pdf}

\usepackage{pgfplots}
\pgfplotsset{compat=1.18}
\usepackage{amsmath}

\newtheorem{lemma}{Lemma}[section]
\newtheorem{theorem}[lemma]{Theorem}
\newtheorem{corollary}[lemma]{Corollary}
\newtheorem{prop}[lemma]{Proposition}

\newcommand{\C}{\mathbb{C}} 
\newcommand{\R}{\mathbb{R}} 
\newcommand{\N}{\mathbb{N}} 

\newcommand{\Rea}{\operatorname{Re}}

\newcommand{\eps}{\varepsilon}

\newcommand{\supp}{\operatorname{supp}}

\newcommand{\sprod}[2]{\langle #1, #2\rangle}

\title{\textbf{Essential self-adjointness of semi-bounded operators}}

\author{M.~Griesemer\footnote{marcel.griesemer@mathematik.uni-stuttgart.de}\,\ and V.~Ku{\ss}maul\footnote{valentin.kussmaul@mathematik.uni-stuttgart.de}\\  
\small Fachbereich Mathematik, Universit\"at Stuttgart, D-70569 Stuttgart, Germany}  
\date{}

\begin{document}
\maketitle

\begin{abstract}
 A semi-bounded operator that is bounded below by one is essentially self-adjoint if the kernel of the adjoint is trivial. 
This is the case if the kernel of the adjoint belongs to the form domain of the Friedrichs extension.
We describe a local version of this criterion, where locality is defined in terms of bounded operators approaching the identity.
The result is an abstract operator theoretic framework that generalizes the method by Wienholtz and Simader developed in the context of semi-bounded elliptic partial differential operators. As applications, we obtain competitive results on essential self-adjointness of many-particle Schr\"odinger operators, pseudo-relativistic Hamiltonians, and the standard model of non-relativistic quantum electrodynamics.
\end{abstract}



\section{Introduction} \label{introduction}

Let $H:\DD\subset \HH\to \HH$ be a densely defined, symmetric operator that is bounded below in the Hilbert space $\HH$. We assume $H\ge 1$ so that essential self-adjointness becomes equivalent to $\ker(H^*)=\{0\}$.
Let $Q(H)$ denote the domain of the closure of the quadratic form defined by $H$. Then $\ker(H^{*})\cap Q(H) = \{0\}$, by \cite{AloSim}, and hence 
\begin{equation}\label{basic}
  \ker(H^*)\subset Q(H)
\end{equation}
implies $\ker(H^*)=\{0\}$. Below we describe an abstract framework in which a local version of \eqref{basic} is still equivalent to essential self-adjointness of $H$. This framework is well adapted to second order elliptic PDEs with locally bounded coefficients, but it is also applicable to certain pseudo-differential operators. We apply it to Schr\"odinger operators with and without magnetic field, to a pseudo-relativistic Hamiltonian and to the Pauli-Fierz Hamiltonian from non-relativistic QED. We thereby recover well-known results from the literature that were obtained previously by a variety of other methods \cite{Kato72,Falconi15,Hiroshima02,HH08}.

To describe the main ideas on which this work is based, we now present a simplified version of our main result, Theorem 2.1. Let $(\chi_n)$ be a sequence of self-adjoint bounded operators in $\HH$ with the following properties:
\begin{itemize}
\item[(a)] $0\le \chi_n \le 1$, $\chi_n = \chi_k\chi_n$ for $k\ge 2n$ and $\chi_n\to 1$ strongly as $n\to\infty$.
\item[(b)]  $\chi_n\DD\subset \DD$, and $\chi_n H\chi_n \le c_n H$ in form sense on $\DD$.
\item[(c)] $\chi_n \ker(H^{*})\subset Q(H)$.
\item[(d)]  $[\chi_n,[\chi_n,H]] : \DD\subset \HH\to\HH$ is a bounded operator and for all $u\in \ker(H^{*})$,
\begin{equation}\label{double-com}
   \liminf_{n\to\infty}\big|\sprod{u}{[\chi_n,[\chi_n,H]]u}\big| = 0.
\end{equation}
\end{itemize}

\begin{theorem}\label{thm1}
If Hypotheses (a)-(d) are satisfied, 
then $H$ is essentially self-adjoint on $\DD$.
\end{theorem}

If we choose $\chi_n=1$ for all $n\in \N$, we see that (a), (b), and (d) are trivially satisfied and (c) reduces to \eqref{basic}. This again shows that \eqref{basic} implies essential self-adjointness of $H$.
Hypothesis (b) is a mild technical prerequisite, which is equivalent to $\chi_n Q(H)\subset Q(H)$. 
Hypotheses (c) and (d) are the main assumptions.  None of them can be dropped from the hypotheses of \Cref{thm1} as can be seen from the counterexamples in the appendix. 

Our method of proof is an operator theoretic version of a method due to Wienholtz and Simader that was invented
for proving essential self-adjointness of semi-bounded second order elliptic PDE's \cite{Wien, Simader78}. In our abstract setting,
the argument is as follows: We claim that
\begin{equation}\label{IMS-form}
      \sprod{\chi_n u }{H \chi_n u} = \Rea\sprod{H^{*}u}{\chi_n^2 u} - \frac{1}{2} \sprod{u}{[\chi_n,[\chi_n,H]]u},
\end{equation}
for all $u\in D(H^{*})$ with $\chi_n u\in Q(H)$ provided (a) and (b) are satisfied and $[\chi_n,[\chi_n,H]]$ is bounded.
\Cref{IMS-form} is a form version of the well-known commutator identity $2\chi_n H \chi_n = \chi_n^2 H + H\chi_n^2 - [\chi_n,[\chi_n,H]]$, which is the IMS-formula in the case of Schr\"odinger operators.
Hence if (a)-(d) hold and  $u\in \ker H^{*}$, then \eqref{IMS-form} in combination with $H\ge 1$ implies
\begin{equation}\label{IMS-ineq}
      \|\chi_n u\|^2 \le -\frac{1}{2} \sprod{u}{[\chi_n,[\chi_n,H]]u}.
\end{equation}
In view of $\chi_n u\to u$ and \eqref{double-com}, we conclude that $u=0$. This proves that $\ker H^{*} = \{0\}$ and hence that $H$ is essentially self-adjoint.

The advantage of this method over other methods is that only local conditions are imposed, besides the requirement that the operator is bounded below. In the case of semi-bounded Schr\"odinger operators, $H = -\Delta +V$ in $L^2(\R^d)$ with $V\in L^2_{\rm loc}(\R^d)$ and $d\le 3$, we may choose $\DD=C_0^{\infty}(\R^d)$ and $\chi_n$ a smooth space cutoff with growing support and vanishing derivative as $n\to\infty$. Then (a) holds by construction of $\chi_n$, (d) follows from
$$
     - [\chi_n,[\chi_n,H]] = |\nabla\chi_n|^2 \to 0\qquad (n\to\infty),
$$
and (b) is easily derived by an IMS localization argument using the partition of unity $\chi_n^2+\zeta_n^2=1$ where $\zeta_n:=\sqrt{1-\chi_n^2}$, see \Cref{abstract section}. For the proof of (c) we compare $H$ with truncated operators $H_n$ that are essentially self-adjoint on $\DD$. Then, by \Cref{thm1}, $H$ is essentially self-adjoint. More generally, if $H = -\Delta +V$ in $L^2(\R^{Nd})$ describes an $N$-particle system in $d\le 3$ space dimensions, then for essential self-adjointness it suffices that all two-body potentials belong to $L^2_{\rm loc}(\R^d)$ and that $H$ is bounded below, see \Cref{mag-schr-no-spin}. In the same vein \Cref{abstract theorem} gives essential self-adjointness for a large class of other semi-bounded Hamiltonians from quantum mechanics and quantum field theory.

In the case of semi-bounded elliptic partial differential operators of second order, essential self-adjointness was first derived from an inequality of the type \eqref{IMS-ineq} by Wienholtz \cite{Wien}. His ideas were later refined in \cite{Stetkaer-Hansen66, Simader78} and generalized to operators on manifolds in \cite{Shubin01,Grummt-Kolb}. Simplified expositions of the method can be found in \cite{Glazman, Borthwick, Helffer}. The argument above for essential self-adjointness of Schr\"odinger operators is due to Simader \cite{Simader78}.

This paper is organized as follows. In \Cref{abstract section} we describe and prove \Cref{abstract theorem}, a generalization of \Cref{thm1}, where the double commutator in Hypothesis (d), above, may be an unbounded operator. Section 3 contains our applications to Schr\"odinger operators, including pseudo-relativistic ones, and to Schr\"odinger operators with magnetic field and with spin. Section 4 is devoted to the essential self-adjointness of Pauli-Fierz Hamiltonians. \Cref{abstract theorem}, with 
$\chi_n$ being a number cutoff, allows us to recover the corresponding results in  \cite{Falconi15,Hiroshima02,HH08} for a slightly larger class of potentials.
There is an appendix containing auxiliary results and counterexamples showing that neither (c) nor (d) may be dropped in \Cref{thm1}.


\section{The abstract argument} 
\label{abstract section}

Let $\HH$ be a separable complex Hilbert space. Let $H:\DD\subset \HH\to \HH$ be densely defined, symmetric and bounded below. 
Let $H_0$ be self-adjoint in $\HH$, bounded below, and assume that $\DD$ is a form core for $H_0$. Upon adding a constant, we may assume that $H, H_0 \geq 1$. 
Let $(\chi_n)_{n \in \N}$ be a sequence of bounded operators in $\HH$ with the following properties: 
\begin{itemize}
\item[(a)] For all $n \in \N$, $0\le \chi_n \le 1$, $\chi_n = \chi_{2 k} \chi_n$ for $k\ge n$, and $\chi_n\to 1$ strongly as $n\to\infty$.
\item[(b)]  For all $n\in \N$, $\chi_n\DD\subset \DD$, and there exists $c_n > 0$ such that in form sense
\begin{align} \label{H_0 H_0 form bound}
	\chi_n H_0 \chi_n \leq c_n H_0 \quad \mathrm{on} \: \,  \DD.
\end{align}
\item[(c)] For all $n \in \N$, there exists $c_n > 0$ such that 
\begin{align}
	\chi_n H \chi_n &\leq c_n H_0 \quad \mathrm{on} \: \, \DD, \label{H H_0 form bound} \\ 
	\| [\chi_n, [\chi_n, H ]] \eta \|& \leq c_n (\| H_0^{1/2}\chi_{2 n}\eta\|+ \| \eta \|) \quad (\eta \in \DD). \label{double comm bound}
\end{align}
\item[(d)]  For all $n \in \N$, 
\begin{align} \label{local regularity}
	\chi_n \ker H^* \subset Q(H_0) 
\end{align}
and for all $u \in \ker H^*$ 
\begin{align} \label{double comm vanishes}
	\liminf_{n\to\infty}\big|\sprod{u}{\overline{[\chi_n,[\chi_n,H]]}u}\big| = 0.
\end{align}
\end{itemize}

Notice that the double commutator $[\chi_n, [\chi_n, H]]$ is symmetric on $\DD$ and hence closable. Conditions \eqref{double comm bound} and \eqref{local regularity} imply that \eqref{double comm vanishes} is well-defined for $u \in \ker H^*$. 

\begin{theorem} \label{abstract theorem}
Suppose the semi-bounded operators $H$ and $H_0$ introduced above satisfy Hypotheses (a) - (d). 
Then $H$ is essentially self-adjoint on $\DD$. 
\end{theorem}

\noindent
\emph{Remarks.} 
\begin{enumerate}
\item  \Cref{abstract theorem} implies \Cref{thm1}. To see this let $H_0$ be the Friedrichs extension of $H$. Then the simplified set of hypotheses in the introduction is obviously sufficient for (a)-(d) above.
\item The main improvement of \Cref{abstract theorem} over  \Cref{thm1} is that $[\chi_n, [\chi_n, H ]]$ may be unbounded, which is important in our application to QFT models. Furthermore, the freedom in the choice of $H_0$ simplifies the verification of the other assumptions. For example, if $H_0$ is known to be essentially self-adjoint on $\DD$, then \eqref{local regularity} easily follows from suitable operator bounds, c.f. \Cref{abstract lemma} with the choice $\hat{H} = H_0$. 
\end{enumerate}

\begin{proof}
Since $H \geq 1$ it suffices to show that $\ker H^* = \{0 \}$. Let $u \in \ker H^*$. By (a)-(c) and \eqref{local regularity}, \Cref{main lemma}, below, shows that for all $n \in \N$,
\begin{align}\label{un-bound}
	\| \chi_n u \|^2 \leq \frac{1}{2} |\sprod{u}{{\overline{[\chi_n,[\chi_n,H]]}u}}|.
\end{align}
Since $\chi_n u \to u$ as $n \to \infty$, it follows from \eqref{un-bound} and \eqref{double comm vanishes} that
\begin{align*}
	\| u \|^2 = \lim_{n \to \infty}\| \chi_n u \|^2 &\leq \liminf_{n \to \infty} \frac{1}{2} |\sprod{u}{{\overline{[\chi_n,[\chi_n,H]]}u}}| = 0. \qedhere
\end{align*}
\end{proof}
\begin{lemma} \label{approx lemma}
Assume (a) and (b). Suppose $u \in \HH$ with $\chi_n u \in Q(H_0)$ for all $n \in \N$. Then there exists a sequence $(u_k)_{k \in \N}$ in $\DD$ such that 
\begin{align} \label{approx 1}
	u_k \to u \quad (k \to \infty)
\end{align}
and, for all $n \in \N$,
\begin{align} \label{approx 2}
	H_0^{1/2}\chi_n u_k \to H_0^{1/2}\chi_n u \quad (k \to \infty).
\end{align}
\end{lemma}

\begin{proof}
For each $k \in \N$, by assumption on $u$, $\chi_{2 k} u \in Q(H_0) = D(H_0^{1/2})$. Since $\DD$ is a form core of $H_0$, 
hence an operator core of $H_0^{1/2}$, there exists $u_k \in \DD$ with 
\begin{align} \label{def approx}
\|H_0^{1/2} (u_k - \chi_{2 k} u) \|< \frac{1}{k},
\end{align}
and in particular, $u_k - \chi_{2k}u \to 0$. In view of $\chi_{2k}u \to u$ we conclude that \eqref{approx 1} holds.
For the proof of \eqref{approx 2} we first notice that \eqref{H_0 H_0 form bound} extends to $Q(H_0)$.
For $k \geq n$ we get from $\chi_n = \chi_n \chi_{2 k}$, \eqref{H_0 H_0 form bound} and  \eqref{def approx} combined,
\begin{align*}
	\|H_0^{1/2}\chi_n u_k - H_0^{1/2}\chi_n u \| &= \| H_0^{1/2}\chi_n (u_k - \chi_{2 k} u) \| \\
	&\leq c_n^{1/2}\|H_0^{1/2}(u_k - \chi_{2 k}u)\| \to 0 \quad (k \to \infty). \qedhere
\end{align*}
\end{proof}

\begin{lemma} \label{main lemma}
Assume (a)-(c) and suppose $u \in D(H^*)$ with $\chi_n u \in Q(H_0)$ for all $n \in \N$.  Then 
\begin{align} \label{double comm ineq}
	\| \chi_n u \|^2 \leq \mathrm{Re}\sprod{H^*u}{\chi_n^2 u} -\frac{1}{2} \sprod{u}{{\overline{[\chi_n,[\chi_n,H]]}u}}.
\end{align}
\end{lemma}

\begin{proof}
Let $(u_k)$ in $\DD$ be the sequence given by \Cref{approx lemma}, which approximates
$u$ in the sense \eqref{approx 1} and \eqref{approx 2}. Then $(H_0^{1/2}\chi_{2n}u_k)_{k\in \N}$ is a Cauchy sequence and hence
we conclude, using \eqref{double comm bound}, that $([\chi_n, [\chi_n, H ]] u_k)_{k \in \N}$ is a Cauchy sequence as well. 
Since $u_k \to u$ as $k \to \infty$ it follows that for all $n \in \N$
\begin{align} \label{double comm conv}
[\chi_n, [\chi_n, H ]] u_k \to \overline{[\chi_n, [\chi_n, H ]]} u \quad (k \to \infty).
\end{align}
As a further preparation, we remark that, for all $n \in \N$,
\begin{align} \label{prod vanishes}
\sprod{u_k}{H \chi_n^2 u_k} \to \sprod{H^* u}{\chi_n^2 u} \quad (k \to \infty),
\end{align}
which we prove below. Now, the commutator identity
\begin{align} \label{IMS-formula}
 \chi_n H \chi_n = \frac{1}{2} (H \chi_n^2 + \chi_n^2 H) - \frac{1}{2} [\chi_n, [\chi_n, H]]
\end{align}
combined with $H\ge 1$ shows that
\begin{align*}
\| \chi_n u_k \|^2 &\leq \sprod{\chi_n u_k}{H \chi_n u_k} \\ 
&= \mathrm{Re} \sprod{u_k}{H \chi_n^2 u_k} - \frac{1}{2} \sprod{u_k}{[\chi_n,[\chi_n,H]]u_k}.
\end{align*}
In view of \eqref{approx 1}, \eqref{double comm conv} and \eqref{prod vanishes}, the desired inequality \eqref{double comm ineq} follows in the limit $k \to \infty$. 

It remains to verify \eqref{prod vanishes}. We have
\begin{align*}
	\sprod{u_k}{H \chi_n^2 u_k} - \sprod{H^* u}{\chi_n^2 u} &= \sprod{u_k -u}{H \chi_n^2 u_k} + \sprod{H^* u}{\chi_n^2 (u_k - u)}.
\end{align*}
By \eqref{approx 1}, the second term on the right-hand side vanishes as $k \to \infty$. Concerning the first term we use \eqref{IMS-formula} in the form $H \chi_n^2 = 2 \chi_n H \chi_n - \chi_n^2 H + [\chi_n, [\chi_n, H]]$. Since $\chi_n = \chi_{2n} \chi_n$ it follows that
\begin{align*}
H \chi_n^2 = \chi_{2 n}H \chi_n^2 + (1 - \chi_{2n}) [\chi_n,[\chi_n, H]]
\end{align*}
and therefore
\begin{align*}
\lefteqn{\sprod{u_k - u}{H \chi_n^2 u_k}} \\ 
&= \sprod{\chi_{2n} (u_k - u)}{H \chi_n^2 u_k} + \sprod{(1-\chi_{2n})(u_k - u)}{ [\chi_n,[\chi_n, H]] u_k}.
\end{align*}
The first term on the right-hand side vanishes as $k \to \infty$ because of the form bound \eqref{H H_0 form bound} and the convergence \eqref{approx 2}. The second term vanishes because of \eqref{approx 1} and \eqref{double comm conv}. 
\end{proof}

The following proposition is our tool for verifying condition \eqref{local regularity} of Hypothesis (d), above, or Hypothesis (c) in \Cref{introduction}.

\begin{prop} \label{abstract lemma}
Suppose $H:\DD\subset \HH\to\HH$ is symmetric and bounded below. Let $\chi$ be a bounded self-adjoint operator with $\chi \, \DD \subset \DD$. Let $\hat{H}:\DD\subset \HH\to\HH$ be essentially self-adjoint and bounded below. 
If there exist $c, d > 0$ such that for all $\eta \in \DD$
\begin{align}
	\| (H - \hat{H}) \, \chi \, \eta \|^2 &\leq c \sprod{\eta}{\hat{H} \eta} + d \| \eta \|^2, \label{relative bound 1}\\ 
	\| [\chi, \hat{H}] \eta \|^2 &\leq c \sprod{\eta}{\hat{H} \eta} + d \|\eta \|^2, \label{relative bound 2}
\end{align}
then $\chi \, D(H^{*})\subset Q(\hat{H})$.
\end{prop}

\begin{proof}
Without loss of generality, we assume $\hat{H} \geq 1$.  The form domain $Q(\hat{H})$ endowed with the inner product 
$v, w \mapsto \sprod{\hat{H}^{1/2} v}{\hat{H}^{1/2} w}$ is a Hilbert space and $\DD\subset Q(\hat{H})$ is dense. 
For $u\in D(H^{*})$ and $\eta \in  \DD$ we have
\begin{align}
\sprod{\chi \, u}{\hat{H}\eta} &= \sprod{u}{\chi \,  \hat{H} \eta} \nonumber\\
&= \sprod{u}{[\chi, \hat{H}]\eta} + \sprod{u}{\hat{H} \, \chi \, \eta} \nonumber \\
&= \sprod{u}{[\chi, \hat{H}]\eta} + \sprod{H^{*}u}{\chi \,  \eta} - \sprod{u}{(H - \hat{H})\chi \,  \eta}. \label{functional in Q(H_0)}
\end{align}
We regard the expression on the right of \eqref{functional in Q(H_0)} as a linear functional of $\eta\in \DD$, which is densely defined in $Q(\hat{H})$.
In view of \eqref{relative bound 1}, \eqref{relative bound 2} this linear function is bounded in the Hilbert space $Q(\hat{H})$. So there exists a vector
$w\in Q(\hat{H})$ such that 
$$\sprod{\chi \, u}{\hat{H} \eta} = \sprod{\hat{H}^{1/2} w}{\hat{H}^{1/2} \eta} = \sprod{w}{\hat{H} \eta}$$
or, equivalently, $\sprod{\chi \,  u - w}{\hat{H} \eta} = 0$ for all $\eta \in \DD$. Since $\hat{H} \geq 1$ and since $\DD$ is a core of $\hat{H}$, it follows that $\hat{H} \DD$ is dense in $\HH$, and hence $\chi \,  u = w\in Q(\hat{H})$.
\end{proof}

The following result is a corollary of \Cref{thm1} and \Cref{abstract lemma}. 

\begin{corollary} \label{abstract corollary}
Let $H: \DD \subset \HH \to \HH$ be symmetric and bounded below. Suppose the following is true:
\begin{enumerate}
\item[(i)] There exist bounded operators $\chi_n$ and $\zeta_n$ in $\HH$ satisfying $0 \leq \chi_n, \zeta_n \leq 1$, $\chi_n^2 + \zeta_n^2 = 1$, $\chi_n = \chi_{2 k} \chi_n$ for $k\geq n$ and $\chi_n \to 1$ strongly as $n \to \infty$.
\item[(ii)] The operators $\chi_n$ and $\zeta_n$ leave $\DD$ invariant, $[\chi_n, [\chi_n, H]]$ and $[\zeta_n, [\zeta_n, H]]$ extend from $\DD$ to bounded operators in $\HH$, and $\| [\chi_n, [\chi_n, H]] \| \to 0$ as $n \to \infty$.
\item[(iii)] For each $n \in \N$ there exists an operators $H_n: \DD \subset \HH \to \HH$  that is essentially self-adjoint and bounded below, s.t. on $\DD$ we have $(H - H_n) \chi_n = 0$, $[\zeta_n, [\zeta_n, H_n]]$ is bounded, and for some constants $c_n, d_n > 0$
\begin{align} \label{first-comm-relative-bound}
	\| [\chi_n, H_{2n}] \eta \|^2 \leq c_n \sprod{\eta}{H_{2n} \eta} + d_n \| \eta \|^2 \quad (\eta \in \DD).
\end{align}
\end{enumerate}
Then $H$ is essentially self-adjoint on $\DD$. 
\end{corollary} 

\begin{proof}
We apply \Cref{thm1}. Without loss of generality, we assume $H \geq 1$. First, we show that conditions (a), (b) and (d) of \Cref{thm1} follow from assumptions (i) and (ii). For (a) and (d) this is obvious. For (b) we use the algebraic identity 
\begin{align}
	H &= \frac{1}{2}\left((\chi_n^2 + \zeta_n^2) H + H (\chi_n^2 + \zeta_n^2) \right) \nonumber \\ 
	&= \chi_n H \chi_n + \zeta_n H \zeta_n + [\chi_n, [\chi_n, H]] + [\zeta_n, [\zeta_n, H]]. \label{IMS-localization}
\end{align}
Since $\zeta_n H \zeta_n \geq \zeta_n^2 \geq 0$ and both commutators are bounded, the form bound (b) follows. 

Second, we derive (c) from (ii) and (iii): From $$(H - H_{2n}) \chi_n = (H - H_{2n}) \chi_{2n} \chi_n = 0$$ and assumption \eqref{first-comm-relative-bound}
we see that \eqref{relative bound 1} and  \eqref{relative bound 2} of  \Cref{abstract lemma} are satisfied for the choices $\chi = \chi_n$ and $\hat{H} = H_{2n}$. This yields
\begin{align} \label{technical-inclusion-H-2n}
\chi_n D(H^*) \subset Q(H_{2n}).
\end{align}
Now it remains to show that 
\begin{equation}\label{technical-inclusion-H}
   \chi_{2n} Q(H_{2n}) \subset Q(H).
\end{equation}
From \eqref{technical-inclusion-H-2n}, \eqref{technical-inclusion-H} and $\chi_{2n} \chi_n =\chi_n$ it follows that 
$\chi_n D(H^*) \subset Q(H)$, which implies (c). 

To prove \eqref{technical-inclusion-H} we argue as above and use that, by (ii) and (iii) the commutators $[\chi_{2n}, [\chi_{2n}, H_{2n}]] = [\chi_{2n}, [\chi_{2n}, H]]$ and $[\zeta_{2n}, [\zeta_{2n}, H_{2n}]]$ are bounded.
Since, moreover, $H_{2n}$ is bounded below we conclude
\begin{align*}
	H_{2n} &=  \chi_{2n} H_{2n} \chi_{2n} + \zeta_{2n} H_{2n} \zeta_{2n} + [\chi_{2n}, [\chi_{2n}, H_{2n}]] + [\zeta_{2n}, [\zeta_{2n}, H_{2n}]] \\ 
	&\geq \chi_{2n} H_{2n} \chi_{2n} - C_n\\ 
	&= \chi_{2n} H \chi_{2n} - C_n,
\end{align*}
where $(H - H_{2n}) \chi_{2n} = 0$ was used in the last line. This bound implies \eqref{technical-inclusion-H}, which ends the proof.
\end{proof}


\section{Schr\"odinger operators}
 \label{sec:schrodinger}

In this section, we apply \Cref{abstract corollary} to Schr\"odinger operators with magnetic field, possibly including spin and statistics, and to pseudo-relativistic Hamiltonians with confining potentials. It is clear that, with the same methods, larger classes of elliptic differential and pseudo-differential operators could be treated.

The localization operators $\chi_n$ and $\zeta_n$ are constructed by choosing smooth functions $\chi, \zeta: \R \to [0,1]$ with $\chi^2 + \zeta^2 = 1$ and $\chi(r) = 1$ for $r \leq 1$ and $\chi(r) = 0$ for $r \geq 2$. We then define for $x \in \R^d$ (or $x \in \R^{d N}$)
\begin{align} \label{spatial-cutoff}
\chi_n(x) = \chi(|x|/n), \quad  \zeta_n(x) = \zeta(|x|/n).
\end{align}

\subsection{Magnetic Schr\"odinger operators} \label{sec:mag-schr}

The following theorem improves on a wide variety of results on essential self-adjointness of semi-bounded magnetic Schr\"odinger operators in flat space \cite{Shubin01,Stetkaer-Hansen66,Simader78,Wien}.
 For background, we recall that a semi-bounded Schr\"odinger operator with potential $V \in L_\mathrm{loc}^2(\R^d)$ need not be essentially self-adjoint on $C_0^\infty(\R^d)$, see Example 4 in Chapter X.2 of \cite{RS75}.
 
\begin{theorem} \label{LS-generalized}
Let $d \geq 1$. Suppose $A \in L_\mathrm{loc}^4(\R^d; \, \R^d)$ with $\mathrm{div}A \in L^2_\mathrm{loc}(\R^d)$ in the distributional sense. Let $V \in L_\mathrm{loc}^2(\R^d)$ be real-valued. Let $V_- = \max(-V, 0)$ and assume that for compact $K \subset \R^d$ the operator $V_- \chi_K$ is $\Delta$-bounded with relative bound less than one. If
\begin{align*}
	H = - \Delta - 2i A \cdot \nabla -i \mathrm{div}A + A^2 + V \quad \text{in} \: \, L^2(\R^d)
\end{align*}
is bounded below on $C_0^\infty(\R^d)$, then it is essentially self-adjoint on $C_0^\infty(\R^d)$.
\end{theorem}

\noindent
\emph{Remarks.}
\begin{enumerate}
\item The quadratic form of $H$ on $C_0^\infty(\R^d)$ is given by 
\begin{equation}\label{H-form}
     \|(-i\nabla + A)\eta\|^2 +\sprod{\eta}{V\eta}, 
\end{equation}
which, by the diamagnetic inequality, is bounded below by $\|\nabla|\eta|\|^2 +\sprod{|\eta|}{V|\eta|}$. 
\item If $V$ is continuous, or satisfies a local Stummel-condition, and $A$ is sufficiently smooth, the result is known from \cite{Povzner, Wien,Stetkaer-Hansen66}. 
\item Our local condition on $V_-$ is satisfied if $d \leq 3$, see \Cref{Sobolev} (i). Sufficient conditions in higher dimensions are $V_- \in L^p_\mathrm{loc}(\R^d)$ with $p > 2$ if $d = 4$ and $p = d/2$ if $d \geq 5$. Under these local integrability conditions on $V_-$ and if $A \in C^1(\R^d; \, \R^d)$, \Cref{LS-generalized} is due to \cite{Shubin01}. Another comparable result can be found in \cite{Grummt-Kolb}, where essential self-adjointness of magnetic Sch\"odinger operators on Riemannian manifolds is established. There the condition on $V_{-}$ is expressed in terms of the local Kato-class, which implies that $V_-$ is locally infinitesimally \textit{form} bounded w.r.t. $-\Delta$, while \Cref{LS-generalized} assumes that $V_{-}$  is locally \textit{operator} bounded w.r.t. $-\Delta$ with relative bound less than one. In dimensions $d \geq 5$ this describes a different class of potentials. For example, the potential $V_{-}(x) = \eps|x|^{-2}$ is not in the local Kato-class for any $\eps > 0$ but it is locally operator bounded with respect to $\Delta$ with bound less than one for sufficiently small $\eps$.
\item If we assume that $V_-$ is \emph{globally} (rather than locally) $\Delta$-bounded with relative bound less than one, then $-\Delta+V$ is bounded below and hence, by the first remark, $H$ is bounded below.
In this case the result agrees with Theorem 3 in \cite{LS81}. The point of \Cref{LS-generalized} is that $H$ can be bounded below while $-\Delta+V$ is not \cite{Shubin01}. 
\item Our proof of \Cref{LS-generalized} is based on \Cref{abstract corollary} in combination with Theorem 3 in \cite{LS81}. 
\end{enumerate} 

\begin{proof}
We apply \Cref{abstract corollary} with $\chi_n$ defined by \eqref{spatial-cutoff}. By the assumption on $A$ and $V$, $H$ is well-defined and symmetric on 
$\DD = C_0^\infty(\R^{d})$. Condition (i) is obviously satisfied and (ii) follows from the fact that for smooth functions $\varphi: \R^d \to \R$, 
\begin{align} \label{double-comm-schr}
	[\varphi, [\varphi, H]] = - |\nabla \varphi|^2 \quad \text{on} \: \, \DD.
\end{align}
Concerning (iii), we define
\begin{align*}
	H_0 &=  - \Delta - 2i A \cdot \nabla -i \mathrm{div}A + A^2 + V_+ \\ 
	H_n &= H_0 - V_- \chi_{2n} \quad (n \in \N)
\end{align*}
where $V_\pm = \mathrm{max}(\pm V, 0)$. By Theorem 3 in \cite{LS81}, $H_n$ is essentially self-adjoint and bounded below on $\DD$. From $H-H_n=V_{-}(\chi_{2n}-1)$ and $\chi_n = \chi_{2n}\chi_n$ 
it follows that $(H - H_n) \chi_n = 0$ on $\DD$.
The bound \eqref{first-comm-relative-bound} for the operator $H_{2n}$ is the only non-obvious condition remaining. We have
\begin{align*}
	[\chi_n, H_{2n}] = 2 i (\nabla \chi_n) \cdot (-i \nabla + A) + (\Delta \chi_n)
\end{align*}
which is bounded w.r.t. the form norm of $H_0$, c.f. \eqref{H-form}. Since the form norms of $H_0$ and $H_{2n}$ are equivalent, \eqref{first-comm-relative-bound} follows. To see the equivalence of the form norms notice that $V_- \chi_{4n}$ is form bounded w.r.t. $H_0$ with relative bound less than one. This follows from the diamagnetic inequality $|\nabla |\eta|| \leq |(\nabla + i A)\eta|$ for $\eta \in \DD$ and from the fact that $V_- \chi_{4n}$ is operator- hence form-bounded w.r.t. $-\Delta$ with relative bound less than one.
\end{proof}


The following theorem spells out the hypotheses of  \Cref{LS-generalized} in the case of a many-body Hamiltonian in $L^2(\R^{Nd})$ with $d \leq 3$. It follows from \Cref{LS-generalized}, but 
the many-body structure and the restriction $d\leq 3$, which is legitimate, physically, allow us to give a self-contained proof based on \Cref{abstract corollary}. 

\begin{theorem} \label{mag-schr-no-spin}
Let $d \leq 3$, $N \geq 1$, and moreover:
\begin{enumerate}
\item Let $v,w \in L_\mathrm{loc}^2(\R^d)$ be real-valued and 
\begin{align*}
V(x) = \sum_{j = 1}^N v(x_j) + \sum_{j < k} w(x_j - x_k) \quad (x \in \R^{dN}).
\end{align*}
\item Let $a \in L_\mathrm{loc}^4(\R^d; \, \R^d)$ with $\mathrm{div} \, a \in L_\mathrm{loc}^2(\R^d)$ in the distributional sense and 
\begin{align*}
	A(x) = (a(x_1), \dots a(x_N)) \quad  (x \in \R^{dN}). 
\end{align*}
\end{enumerate}
If the operator
\begin{align} \label{mag-schr}
	H = - \Delta - 2i A \cdot \nabla -i \mathrm{div}A + A^2 + V \quad \mathrm{in} \: \, L^2(\R^{d N})
\end{align}
is bounded below on $C_0^\infty(\R^{d N})$, then $H$ is essentially self-adjoint on $C_0^\infty(\R^{d N})$.
\end{theorem}

\noindent
\emph{Remark.} 
By the first remark following \Cref{LS-generalized}, $H$ is bounded below if $-\Delta +V$ is bounded below, which is the case 
if the negative parts of $v,w$ belong to $L^2(\R^d) + L^\infty(\R^d)$. This follows from the Kato-Rellich theorem and \Cref{Sobolev} (i).

 \begin{proof}
We use \Cref{abstract corollary} with cutoff functions $\chi_n$ defined by \eqref{spatial-cutoff}. 
By the assumptions on $A$ and $V$, 
$H$ is well-defined and symmetric on $\DD=C_0^\infty(\R^{dN})$.
Condition (i) is obviously satisfied and (ii) follows from \eqref{double-comm-schr}. As for (iii) we define
\begin{equation} \label{mag-schr-no-spin-short-aux}
          H_n = -\Delta - 2i A_n \cdot \nabla -i \mathrm{div} A_n + A_n^2 + V_n 
\end{equation}
where $A_n = A \chi_{2n}$ and $V_n = V \chi_{2n}$. Like $H$, $H_n$ is well-defined and symmetric on $\DD$.
To show that $H_n$ is self-adjoint on $D(-\Delta)$, essentially self-adjoint on $\DD$, and bounded below, it suffices, by Kato-Rellich, to verify that the operators $A_n^2, V_n, \mathrm{div} A_n$ and $A_n \cdot \nabla$ are infinitesimally $\Delta$-bounded. Since $\chi_{2n}$ is compactly supported, the operators $A_n^2, V_n$ and 
\begin{align} \label{div-product-rule-A-n}
\mathrm{div} A_n = (\nabla \chi_{2n}) \cdot A + \chi_{2n} \,  \mathrm{div} A
\end{align}
can be estimated in terms of sums of functions which are infinitesimally $\Delta$-bounded by \Cref{Sobolev} (i). Similarly, each component of $A_n$ is infinitesimally $\nabla$-bounded by \Cref{Sobolev} (ii), so $A_n \cdot \nabla$ is infinitesimally $\Delta$-bounded. From \eqref{mag-schr-no-spin-short-aux}, \eqref{div-product-rule-A-n}, $\chi_n = \chi_{2n}\chi_n$ and $(\nabla \chi_{2n})\chi_n = 0$ it follows that  $(H - H_n) \chi_n = 0$. By \eqref{double-comm-schr}, $[\zeta_n, [\zeta_n, H_n]] = - |\nabla \zeta_n|^2$. It remains to verify \eqref{first-comm-relative-bound}. We have
\begin{align} \label{first-comm-schr-aux}
[\chi_n, H_{2n}] = 2 i (\nabla \chi_n) \cdot (-i \nabla + A_{2n}) + (\Delta \chi_n).
\end{align}
Since $A_{2n}$ is $\nabla$-bounded, \eqref{first-comm-schr-aux} is bounded w.r.t. the form norm of $-\Delta$. Since $H_{2n}$ is a small perturbation of $-\Delta$, the form norms of $-\Delta$ and $H_{2n}$ are equivalent and hence \eqref{first-comm-schr-aux} is form-bounded w.r.t.  $H_{2n}$.
In summary, all conditions of \Cref{abstract corollary} are satisfied and the theorem is thus proven.
 \end{proof}

The next theorem is a variant of \Cref{mag-schr-no-spin}  including spin and statistics. We use the identification
$$
	\otimes^N (L^2(\R^3) \otimes \C^2) \simeq L^2(\R^{3N}; \, \otimes^N \C^2)
$$
and denote the anti-symmetric subspace by 
$$
	\otimes_\mathrm{a}^N (L^2(\R^3) \otimes \C^2) \simeq L_\mathrm{a}^2(\R^{3 N}; \, \otimes^N \C^2).
$$
Let $C_{0, \mathrm{a}}^\infty(\R^{3N}; \, \otimes^N \C^2)$ be the subspace of anti-symmetric test functions. 

\begin{theorem}
Let $N \geq 1$ and moreover:
\begin{enumerate}
\item Let $v,w \in L_\mathrm{loc}^2(\R^3)$ be real-valued with $w(y) = w(-y)$ for a.e. $y \in \R^3$ and 
\begin{align*}
V(x) = \sum_{j = 1}^N v(x_j) + \sum_{j < k} w(x_j - x_k) \quad (x \in \R^{3N}).
\end{align*}
\item Let $a \in L_\mathrm{loc}^4(\R^3; \, \R^3)$ with $\mathrm{div} \, a \in L_\mathrm{loc}^2(\R^3)$ and
\begin{align*}
	A(x) = (a(x_1), \dots a(x_N)) \quad  (x \in \R^{3N}). 
\end{align*}
\item Let $\sigma_j$ denote a triple of self-adjoint matrices in $\C^{2 \times 2}$ acting on the $j$th factor in $\otimes^N \C^2$, let $b \in L_\mathrm{loc}^2(\R^3; \, \R^3)$ and 
\begin{align*}
	\sigma &= (\sigma_1, \dots, \sigma_N), \\ 
	B(x) &= (b(x_1), \dots, b(x_N)) \quad (x \in \R^{3 N}).
\end{align*}
\end{enumerate}
If the operator
\begin{align}\label{mag-schr-spin}
	H =- \Delta - 2i A \cdot \nabla -i \mathrm{div}A + A^2 + V + \sigma \cdot B \quad \mathrm{in} \: \, L_\mathrm{a}^2(\R^{3 N}; \, \otimes^N \C^2)
\end{align}
is bounded below on $C_{0, \mathrm{a}}^\infty(\R^{3N}; \otimes^N \C^2)$, then $H$ is essentially self-adjoint on $C_{0, \mathrm{a}}^\infty(\R^{3N}; \otimes^N \C^2)$. 
\end{theorem}

\noindent
\emph{Remark.} 
By the first remark following \Cref{LS-generalized} and \Cref{Sobolev} (i), a sufficient condition for $H$ to be bounded below is that each component of $b$ is in $L^2(\R^3) + L^\infty(\R^3)$ and that the negative parts of $v$ and $w$ are in $L^2(\R^3) + L^\infty(\R^3)$.
\begin{proof}
The proof is a copy of the proof of \Cref{mag-schr-no-spin}. The term $\sigma \cdot B$ is treated like $V$. 
\end{proof}

\subsection{Pseudo-relativistic Schr\"odinger operators} \label{sec:semi-relativistic}

We now turn to essential self-adjointness for a class of pseudo-differential operators including 
pseudo-relativistic Schr\"odinger operators with confining potentials. For preparation we need:

\begin{lemma} \label{fourier commutator}
Let $\chi \in C_0^\infty(\R^d)$. Let $f: \R^d \to \R$ be Lipschitz continuous with constant $L$. Then $[\chi, f(-i \nabla)]$ extends from Schwartz space to a bounded operator in $L^2(\R^d)$ with
\begin{align*} 
	\|[\chi, f(-i \nabla)]\| \leq \frac{L}{(2\pi)^{d/2}} \int_{\R^d} |\widehat{\nabla \chi}(k)|dk.
\end{align*}
\end{lemma}
\begin{proof}
We have
\begin{align} \label{loc-in-fourier}
\chi(x) = \frac{1}{(2 \pi)^{d/2}} \int_{\R^d} \widehat{\chi}(k) e^{i k \cdot x}dk.
\end{align}
Let $p = -i \nabla$. In Fourier space, the operator $x$ satisfies $(e^{i k \cdot x} \eta)(p) = \eta(p + k)$ so that $[e^{i k \cdot x}, f(p)] = [f(p + k) - f(p)] e^{i k \cdot x}$. Combined with $|f(p + k) - f(p)| \leq L |k|$ and $|\widehat{\chi}(k)| |k| = |\widehat{\nabla \chi}(k)|$ the assertion follows from \eqref{loc-in-fourier}.
\end{proof}

\begin{theorem}
Let $f: \R^d \to \R_+$ be a Lipschitz function. Let $V \in L_\mathrm{loc}^2(\R^d)$ be real-valued and suppose that for each compact $K \subset \R^d$ the operator $V \chi_K$ is $f(-i\nabla)$-bounded
with relative bound smaller than one. If the operator
\begin{align*}
	H = f(-i\nabla) + V \quad \mathrm{in} \: \, L^2(\R^d)
\end{align*}
is bounded below on $C_0^\infty(\R^d)$, then H is essentially self-adjoint on $C_0^\infty(\R^d)$.
 \end{theorem}
 
\noindent
\emph{Remarks.} 
\begin{enumerate}
\item The result is true in particular for  $f(p) = \sqrt{p^2 + m^2}$ with $m\geq 0$. 
\item If the negative part $V_- = \mathrm{max}(0, -V)$ is \emph{globally} $f(-i \nabla)$-bounded with relative bound less than one, then $H$ is bounded below. There is no restriction on the growth of $V_+$.
\end{enumerate}


\begin{proof}
First, the result holds when $V = 0$. Indeed, $f(-i \nabla)$ is essentially self-adjoint on $C_0^\infty(\R^d)$ because $C_0^\infty(\R^d)$ is dense in $H^1(\R^d) = D(|-i\nabla|)$ and, by the Lipschitz continuity, we have the dense embedding $D(|-i \nabla|) \subset D(f(-i \nabla))$.

In the general case, we apply \Cref{abstract corollary} with cutoff operators $\chi_n$ and $\zeta_n$ defined in \eqref{spatial-cutoff}. By \Cref{fourier commutator}, the commutator $[\chi_n, f(-i \nabla)]$ extends from $C_0^\infty(\R^d)$ to a bounded operator with
\begin{align} \label{first comm to 0}
	\|[\chi_n, f(-i \nabla)] \| &\leq  \frac{L}{n (2 \pi)^{d/2}} \int_{\R^d} |\widehat{\nabla \chi_1}(k)| dk 
\end{align}
and, since $\zeta_n - 1 \in C_0^\infty(\R^d)$, the same holds with $\zeta_n$ in place of $\chi_n$. Hence conditions (i) and (ii) of \Cref{abstract corollary} are satisfied. For (iii) we choose operators
\begin{align*}
	H_n = f(-i \nabla) + V \chi_{2n}.
\end{align*}
By assumption and Kato-Rellich, $H_n$ is self-adjoint on the domain of $f(-i \nabla)$ and essentially self-adjoint on $C_0^\infty(\R^d)$. The assumptions on the commutators in (iii) follow from the boundedness of \eqref{first comm to 0}.
\end{proof}


\section{Essential self-adjointness of Pauli-Fierz Hamiltonians}
\label{PF-ess-sa-section}

In this section we use \Cref{abstract theorem} to prove essential self-adjointness of Pauli-Fierz Hamiltonians with an electrostatic potential $V$ whose positive part $V_{+}$ is an arbitrary, locally square integrable real functions. 
The results on essential self-adjointness from \cite{BFS99,Falconi15,Hiroshima02,HH08} emerge as special cases. 


Let $\hh$ be any separable Hilbert space and let $\FF$ be the symmetric Fock space over $\hh$. 
We consider Hamiltonians in $L^2(\R^d) \otimes \FF$ of the form
 \begin{align} \label{PF hamiltonian}
	H = (-i\nabla + A)^2 + H_f + V,
\end{align} 
which are \emph{generalized} Pauli-Fierz Hamiltonians, in the sense that the one-photon space $\hh$ remains unspecified. This is also reflected in the following fairly weak 
hypotheses: 
\begin{itemize}
	\item[(i)] $V \in L^2_\mathrm{loc}(\R^d)$ is real-valued and $V_- = \max(-V, 0)$ is operator bounded with respect to $-\Delta$ with relative bound smaller than $1$.  
	\item[(ii)] $H_f = d\Gamma(\omega)$ in $\FF$ is the second quantization of a self-adjoint operator $\omega \geq 0$ in $\hh$.
	\item[(iii)] Each component of $A = (A_1, ..., A_d)$ is an operator in $L^2(\R^d) \otimes \FF \stackrel{\sim}{=} L^2(\R^d; \, \FF)$ of the form
	 $$(A_j u)(x) = [a^*(G_j(x)) + a(G_j(x))] u(x)$$
	 with $G_j \in L^\infty(\R^d; \, \hh)$ weakly $\partial_j$-differentiable\footnote{
	 $G_j \in L^\infty(\R^d; \, \hh)$ is weakly $\partial_j$-differentiable if there exists $F \in L^\infty(\R^d; \, \hh)$ such that for all $\eta \in \hh$ and $f \in C_0^\infty(\R^d)$ we have $ \int \sprod{\eta}{F(x)}_\hh f(x) dx = - \int  \sprod{\eta}{G_j(x)}_\hh  \partial_j f(x) dx $.} and $\partial_j G_j \in L^\infty(\R^d; \, \hh)$.
\end{itemize}
Let $\hh_0$ be a core for $\omega$ and 
\begin{align} \label{def-field-core}
	\DD_f = \big\{\psi \in \FF_\mathrm{fin} | \, \psi^{(n)} \in \otimes_\mathrm{sym}^n \hh_0 \big\} 
\end{align}
with $\FF_\mathrm{fin}$ the finite particle subspace of $\FF$. Then $\DD_f$ is a core for $H_f$. 

\begin{theorem} \label{PF main thm}
Assume $\mathrm{(i) - (iii)}$. Then $H$ is essentially self-adjoint on $\DD = C_0^\infty(\R^d) \otimes \DD_f$. 
\end{theorem}

\noindent
\emph{Remarks.} 
\begin{enumerate}
\item Notice the absence of any restrictions on $V_{+}$ other than  $V_{+} \in L^2_\mathrm{loc}(\R^d)$. A comparable result with $V_{-}=0$ is known from Section 4.3 of \cite{Falconi15}. 

\item If $V$ is infinitesimally $\Delta$-bounded and if $H$ denotes the usual Pauli-Fierz Hamiltonian, which is a special case of \eqref{PF hamiltonian}, then $H$ is known to be self-adjoint on the domain of $-\Delta + H_f$ and hence it is essentially self-adjoint on any core of $-\Delta + H_f$ \cite{BFS99,HH08,Hiroshima02,Kussmaul}.  
\Cref{PF main thm}, like the similar result in \cite{Falconi15}, shows that essential self-adjointness holds without restrictions on the infrared part of the electron-photon interaction. The same is true with the coupling of spin and radiation included, which we neglect for notational simplicity.

\item Our proof of \Cref{PF main thm} is based on \Cref{abstract theorem} and on the results on essential self-adjointness of $-\Delta + V$ in \cite{Kato72,RS75}. 
Essential self-adjointness of $H_0\coloneqq H_{A = 0}$ carries over to $H$ in our proof. We use a number cutoff in Fock space for the localization operators $\chi_n$, see \eqref{Fock cutoff}, and ideas from \cite{Falconi15} to verify condition \eqref{double comm vanishes}.  

\item The operator $(-i\nabla + A)^2$ in \eqref{PF hamiltonian} is defined by
$$(-i\nabla + A)^2 = -\Delta -i \nabla \cdot A -i A \cdot \nabla + A^2,$$ which is well-defined on $C_0^\infty(\R^d) \otimes \FF_\mathrm{fin}$. 
By Lemma 13 in \cite{HH08}, $\nabla \cdot A = a^*(\mathrm{div} \, G) + a(\mathrm{div}  \, G) + A \cdot \nabla$.
\end{enumerate}

\noindent
In the following $\chi_n$ denotes the number cutoff
\begin{align} \label{Fock cutoff}
	\chi_n = \chi(N \leq n) \quad (n \in \N) 
\end{align}
with $N = d\Gamma(1)$ the number operator in $\FF$, and 
\begin{align} \label{aux hamiltonian}
	H_0 = -\Delta + V + H_f \quad \mathrm{in} \: \,  L^2(\R^d) \otimes \FF.
\end{align}
As a preparation for the proof of \Cref{PF main thm} we first establish the following two lemmas.

\begin{lemma} \label{basic properties}
Assume (i) - (iii). On $\DD$, $H$ and $H_0$ are bounded below. The operators $-\Delta$, $\chi_n H_0 \chi_n$ and $\chi_n H \chi_n$ are form bounded with respect to $H_0$.
\end{lemma}

\begin{proof}
By assumption (i), $V_-$ is form bounded w.r.t.$-\Delta + V_+$ with relative bound smaller than $1$, so $-\Delta + V$ and $-\Delta + V_+$ are relatively form bounded to each other. From $H_f \geq 0$ it follows that $H_0$ is bounded below and that $-\Delta$ is form bounded w.r.t. $H_0$. By the diamagnetic inequality in QED \cite{GK25}, any lower bound for $-\Delta + V$ is a lower bound for $H$, so $H$ is bounded below.
To bound the form of  $\chi_n H \chi_n$ from above, we drop $V_{-}$ and use that $ A \chi_n$ is a bounded operator. We find
\begin{align*}
\sprod{u}{\chi_n H \chi_n u} &= \| (-i\nabla + A) \chi_n u \|^2 + \sprod{u}{\chi_n V \chi_n u} + \sprod{u}{\chi_n H_f \chi_n u} \\ 
&\leq 2 \| \nabla u \|^2 + 2 \| A \chi_n \|^2 \|u\|^2 + \| V_+^{1/2}u \|^2 +  \| H_f^{1/2}u \|^2 \\ 
&\leq c_n \sprod{u}{(-\Delta + V_+ + H_f + 1) u} \quad (u \in \DD)
\end{align*}
Since $-\Delta + V_{+} + H_f$ is form-bounded w.r.t. $H_0$, the desired bound for $\chi_n H \chi_n$ follows. The bound for $\chi_n H_0 \chi_n$ holds as well 
because $H_0=H$ for $A=0$.
\end{proof}

\begin{lemma} \label{double comm}
	Assume (i) - (iii). Let $\chi_n$ be defined by \eqref{Fock cutoff} and let $\chi^n = \chi(N = n)$. Then, on $\DD$, 
	\begin{align}
		i[\chi_n, H] =  (-i\nabla + A) \,  \cdot  &\left( -i \chi^{n + 1}a^*(G) \chi^{n} + i \chi^n a(G) \chi^{n + 1} \right) \nonumber \\ 
		+ &\left(i \chi^n a(G)\chi^{n + 1} - i \chi^{n + 1} a^*(G) \chi^n \right) \cdot (-i\nabla + A).
		\label{first comm}
	\end{align}
\end{lemma}
\begin{proof}
	This follows from $A = a(G) + a^*(G)$, from
	\begin{align} \label{chi a comm}
		[\chi_n, a(G)] &= \, \chi^n a(G) \chi^{n + 1} \\ 
		[\chi_n, a^*(G)] &= - \chi^{n + 1} a^*(G) \chi^{n}, \label{chi a* comm}
	\end{align}
	and from the fact that $\chi_n$ commutes with $\nabla, V$ and $H_f$. 
\end{proof}

\begin{proof}[Proof of \Cref{PF main thm}]
We apply \Cref{abstract theorem} with cutoff operators $(\chi_n)$ defined by \eqref{Fock cutoff} and $H_0$ defined by \eqref{aux hamiltonian}. By Theorem X.29 in \cite{RS75}, $-\Delta + V$ is essentially self-adjoint on $C_0^\infty(\R^d)$, so $H_0$ is essentially self-adjoint on $\DD = C_0^\infty(\R^d) \otimes \DD_f$.
By \Cref{basic properties}, both $H_0$ and $H$ are bounded below on $\DD$ and hence, upon adding a constant, we may assume $H_0, H \geq 1$. Again by \Cref{basic properties}, conditions \eqref{H_0 H_0 form bound} and \eqref{H H_0 form bound} of \Cref{abstract theorem} are satisfied. It remains to verify \eqref{double comm bound}, \eqref{local regularity} and \eqref{double comm vanishes}.

\medskip\noindent
\underline{\emph{Step 1}}. For each $n \in \N$ there exists $c_n > 0$ such that 
\begin{align} \label{PF first comm bound}
	\|[\chi_n, H] \eta \| \leq c_n ( \| H_0^{1/2} \chi_{2n} \eta\| + \| \eta \|) \quad (\eta \in \DD).
\end{align}

We use expression \eqref{first comm} for the commutator $[\chi_n, H]$ and the product rule 
\begin{align} \label{a prod rule}
	\nabla \cdot  a^{\#}(G) = a^{\#}(\mathrm{div} \, G) + a^{\#}(G) \cdot \nabla, 
\end{align}
where $a^{\#}(G)$ denotes either $a^*(G)$ or $a(G)$. In view of the standard bounds \eqref{relative sqrt N bound}, we arrive at
\begin{align*}
	\|[\chi_n, H] \eta \| \leq c_n (\| \nabla \chi_{2 n} \eta \| + \| \eta \|) \quad (\eta \in \DD).
\end{align*}
Since, by \Cref{basic properties}, the operator $\nabla$ is $H_0^{1/2}$-bounded, \eqref{PF first comm bound} follows.

\medskip\noindent
\underline{\emph{Step 2}}. If $u \in \ker H^*$ then $\chi^n u \in Q(H_0)$ and 
\begin{align} \label{p summable}
\sum_{n = 1}^\infty \frac{1}{n} \|H_0^{1/2} \chi^n u \|^2 < \infty \quad (\chi^n = \chi(N = n)).
\end{align}

The proof of this step is similar to Falconi's arguments in Theorem 3.1 of \cite{Falconi15}. From
\begin{align}
	H - H_0 &=  -i a^*(\mathrm{div} \, G) -i a(\mathrm{div}  \, G) -2 i A \cdot \nabla + A^2 \label{PF interaction}
\end{align}
we see that $(H - H_0) \chi_n$ is $H_0^{1/2}$-bounded. Since $[\chi_n, H_0] = 0$, we conclude from \Cref{abstract lemma} with the choices $\chi = \chi_n$ and $\hat{H} = H_0$ that 
 $\chi_n D(H^*) \subset Q(H_0)$. Hence $\chi_n u\in Q(H_0)$.
It remains to verify \eqref{p summable}. If $u \in \ker H^*$, then for all $\eta \in \DD$,
\begin{align*}
	\sprod{H_0^{1/2}\chi^n u}{H_0^{1/2}  \eta} = \sprod{u}{H_0 \chi^n \eta} = \sprod{u}{(H_0 - H) \chi^n \eta}.
\end{align*}
This extends to all $\eta\in Q(H_0)$ because $\DD$ is a core of $H_0^{1/2}$. Choosing $\eta = \chi^n u$ we get
\begin{align*}
	\| H_0^{1/2} \chi^n u \|^2 &= \sprod{u}{(H_0 - H) \chi^n u} \\ 
	&= O(\sqrt{n}) \sum_{\ell = -2}^2 \| \chi^{n + \ell}u \|  \big( O(\sqrt{n}) \| \chi^n u \| + \| \nabla \chi^{n}u \| \big), 
\end{align*}
where in the last line we estimated the expression $\eqref{PF interaction}$ using $\| a^{\#}(G) \chi^n u \| = O(\sqrt{n}) \| \chi^n u \|$ and the same bound for $\mathrm{div} \, G$. 
Next, we use
\begin{align*}
\| \chi^{n + \ell}u \| \cdot \| \chi^n u \| &\leq \| \chi^{n + \ell}u \|^2 + \| \chi^n u \|^2, \\ 
\| \chi^{n + \ell}u \| \cdot \| \nabla \chi^{n}u \| &\leq \| \chi^{n + \ell}u \|^2/\eps +  \eps \| \nabla \chi^{n}u \|^2,
\end{align*}
combined with $\|\nabla  \chi^n u \|^2 \leq C \| H_0^{1/2} \chi^n u \|^2$, to obtain a bound of the form 
\begin{align*}
	\| H_0^{1/2} \chi^n u \|^2 =  O(n)  \sum_{\ell = -2}^2 \| \chi^{n + \ell}u \|^2.
\end{align*}
This implies \eqref{p summable}.

\medskip\noindent
\underline{\emph{Step 3}}. If $u \in \ker H^*$ then\footnote{In fact, the first commutator $[\chi_n, H]$ alone yields the desired decay. The double commutator $[\chi_n, [\chi_n, H]]$ has no improved behavior.}
 $$\liminf_{n \to \infty} |\sprod{u}{\overline{[\chi_n, [\chi_n, H]]} u}| = 0.$$

Since $\sum_{n} n^{-1} = \infty$ it suffices to show that
\begin{align} \label{liminf sufficient}
	\sum_{n = 1}^\infty \frac{1}{n} |\sprod{u}{\overline{[\chi_n, [\chi_n, H]]}u}| < \infty.
\end{align}
By the expression \eqref{first comm} for $[\chi_n, H]$ and by $\| a^{\#}(G) \chi^n u \| = O(\sqrt{n}) \| \chi^n u \|$,
\begin{align}
	\lefteqn{|\sprod{u}{\overline{[\chi_n, [\chi_n, H]]} u}|}\notag \\ 
	& \quad = 2 |\mathrm{Re}\sprod{\chi_n u}{\overline{[\chi_n, H]} u}| \notag\\ 
	& \quad  = \bigg(O(\sqrt{n}) \sum_{\ell = -2}^{2} \|\chi^{n + \ell} u \| \bigg) \bigg(\| \nabla \chi^n u \| + \|\nabla \chi^{n + 1} u \| +  O(\sqrt{n}) \sum_{\ell = -2}^2 \|\chi^{n + \ell}u\| \bigg)\notag \\ 
	& \quad  \leq \| \nabla \chi^n u \|^2 + \|\nabla \chi^{n + 1} u \|^2 + O(n) \sum_{\ell = -2}^2 \|\chi^{n + \ell}u\|^2.\label{last}
\end{align}
Since $\nabla$ is $H_0^{1/2}$-bounded, \eqref{last} combined with \eqref{p summable} implies \eqref{liminf sufficient}.

With the Conditions \eqref{double comm bound},  \eqref{local regularity}, and \eqref{double comm vanishes} now being established in the Steps 1-3 above, all hypotheses of 
\Cref{abstract theorem} are satisfied and essential self-adjointness of $H$ on $\DD$ follows from \Cref{abstract theorem}.
\end{proof}

\appendix 
\section{Appendix}


\begin{lemma}\label{Sobolev}
Let $d \leq 3$ and $N \geq 1$. Let $v: \R^d \to \R$. Let $V : \R^{dN} \to \R$ be either defined by $V(x) = v(x_j)$ or by $V(x) = v(x_j - x_k)$ for $j \neq k$ where $x = (x_1, \dots x_N) \in \R^{d N}$. 
\begin{enumerate}
\item[(i)] If $v \in L^p(\R^d) + L^\infty(\R^d)$ with $p \geq 2$, then $V$ is infinitesimally $\Delta$-bounded. 
\item[(ii)] If $v \in L^p(\R^d) + L^\infty(\R^d)$ with $p \geq 3$, then $V$ is infinitesimally $\nabla$-bounded.
\end{enumerate}
\end{lemma}

\begin{proof}
By a change of coordinates it suffices to prove the lemma for $V = v: \R^d \to \R$. If $V \in L^p(\R^d) + L^\infty(\R^d)$, then for each $\eps > 0$ there exists a decomposition $V = V_\eps + V^\eps$ with $\|V_\eps\|_p < \eps$ and $\|V^\eps\|_\infty < \infty$. The claims (i) and (ii) now follow H\"older's inequality $\|V_\eps\psi\|_2\le \|V_\eps\|_p \|\psi\|_q$ with $p^{-1}+q^{-1} = 1/2$, from the Sobolev embeddings in dimensions $d \leq 3$
\begin{enumerate}
\item[(i')] $H^2(\R^d) \to L^q(\R^d) \quad (2 \leq q \leq \infty)$,
\item[(ii')] $H^1(\R^d) \to L^q(\R^d) \quad (2 \leq q \leq 6)$,
\end{enumerate}
and from the fact that the graph norms of $\Delta$ and $\nabla$ are equivalent to the $H^2$ and $H^1$ norms, respectively.
\end{proof}

\begin{lemma} \label{standard bounds}
If $m \in \N$ and $h_1, \dots, h_m \in L^\infty(\R^d; \, \hh)$ then 
\begin{align} 
\| a^{\#}(h_1) \dots a^{\#}(h_m) (N + 1)^{-m/2} \| &\leq C_m \, \| h_1 \|_\infty  \dots \| h_m \|_\infty  \label{relative sqrt N bound}
\end{align}
where $N = d\Gamma(1)$ and  $a^{\#} \dots a^{\#}$ denotes any combination of $a$ and $a^*$. 
\end{lemma}
For the proof see Lemma 17 in \cite{FGS01}. 



\subsection*{Counterexamples}

We give two examples here showing that the statement of \Cref{thm1} is wrong in general if either Hypothesis (c) or Hypothesis (d) is eliminated.

\medskip\noindent
\begin{itemize}
\item[(1)] Let $H=-\Delta$ in $\HH = L^2(\R_{+})$ and let $\DD = C_0^{\infty}(\R_{+})$.  Let $\chi_n(x)=\chi(x/n)$ where $\chi\in C^{\infty}(\R;[0,1])$ with $\chi(x)=1$ for $x\le 1$ and $\chi(x)=0$ if $x\ge 2$. Then Hypotheses (a), (b) and (d) are satisfied, but $H$ is not essentially self-adjoint. 
\item[(2)] Let $H=-\Delta$ in $\HH = L^2(I)$, where $I=(0,1)$, and let $\DD = C_0^{\infty}(I)$. Let $\chi_n\in C_0^{\infty}(I)$ with $\chi_n(x) = 1$ for $x\in [1/n,1-1/n]$ and $\supp(\chi_n)\subset [1/2n,1-1/2n]$. Then Hypothesis (a) and (b) are clearly satisfied. To prove (c) we notice that $D(H^{*}) = H^2(I)$. 
It follows that $\chi_n D(H^{*}) \subset H_0^2(I)\subset H_0^1(I) = Q(H)$. But $H$ is not essentially self-adjoint.
 \end{itemize}

\medskip
\noindent
\emph{Acknowledgement.} 
This work is supported by the German Research Foundation (DFG) under Grant GR 3213/4-1.


\end{document}